\documentclass[twocolumn]{svjour3}
\usepackage{graphics}
\usepackage{amssymb}
\usepackage{lineno}
\usepackage{comment}
\usepackage{color}
\usepackage{layouts}
\usepackage{afterpage}
\usepackage{mathrsfs}
\usepackage{amsmath}
\usepackage{natbib}
\usepackage{booktabs}
\usepackage{paralist}
\usepackage{setspace}
\usepackage[utf8]{inputenc}
\usepackage[english]{babel}
\usepackage{mathtools}
\usepackage{centernot}
\usepackage{sidecap}
\usepackage{graphicx,psfrag,epsf}
\usepackage{enumerate}
\usepackage{floatrow}
\usepackage{float}
\usepackage{placeins}

\newcommand*\cpp{C\kern-0.2ex\raisebox{0.4ex}{\scalebox{0.8}{+\kern-0.4ex+}}}
\newcommand{\slfrac}[2]{\left.#1\middle/#2\right.}
\DeclareMathOperator*{\med}{med}
\DeclareMathOperator*{\MAD}{mad}

\DeclareMathOperator*{\ave}{ave}
\DeclareMathOperator*{\SVD}{svd}
\DeclareMathOperator*{\rank}{rank}
\DeclareMathOperator*{\var}{var}
\DeclareMathOperator*{\od}{OD}
\DeclareMathOperator*{\sd}{SD}

\DeclareMathOperator*{\argmin}{argmin}

\DeclareMathOperator{\Tr}{Tr}

\definecolor{ROBPCA}{RGB}{255,0,0}
\definecolor{PcaPP}{RGB}{0,0,255}
\definecolor{PcaL}{RGB}{0,255,0}
\definecolor{HCS}{RGB}{255,0,255}

\begin{document}

%%%%%%%%%%%%%%%%%%%%%%%%%%%%%%%%%%%%%%%%%%%%%%%%%%%%%%%%%%%%%%%%%%%%%%%%%%%%%%

  \title{\bf The FastHCS Algorithm for Robust PCA}
  \author{Eric Schmitt and Kaveh Vakili}
    
    \institute{E. Schmitt \at
              Protix\\
            Industriestraat 3\\
            5107 NC Dongen
              Tel.: +31162782501\\
              \email{eric.schmitt@protix.eu} 
}

  \maketitle

\begin{abstract}
Principal component analysis (PCA) is widely used to analyze high-dimensional data, but it is very sensitive to outliers.
Robust PCA methods seek fits that are unaffected by the outliers and can therefore be trusted to reveal them.
FastHCS (High-dimensional Congruent Subsets) is a robust PCA algorithm suitable for high-dimensional applications, including cases where the number of variables exceeds the number of observations.
After detailing the FastHCS algorithm, we carry out an extensive simulation study and three real
 data applications, the results of which show that FastHCS is systematically more robust to outliers than state-of-the-art methods.
\end{abstract}

\noindent%
{\it Keywords:}  High-dimensional data, outlier detection, computational statistics, exploratory data analysis

\section{Introduction}

Principal component analysis (PCA) is widely used to explore high-dimensional data. 
It centers and rotates the original $p$-dimensional  measurements to construct a small number $q$ of new orthonormal variables, called \emph{principal components}, that  account for most of the variation in the data.
However, classical PCA is very sensitive to outliers. 
Outliers are observations that are inconsistent
with the multivariate pattern of the majority of the data. 
If left unchecked, they
influence the estimated parameters by 
disproportionately pulling the fit towards themselves.
%This obscures the main relationships in the data and their true outlyingness. 
In this way, outliers obscure the main relationships in the data and their true outlyingness.
In practice, we want to find the outliers to bound their 
influence on the fit and to study as objects of interest in their own right.
For these reasons, we need robust PCA methods that meet the following criteria: 
\begin{inparaenum}[(1)]
\item Like classical PCA, a robust PCA method should handle cases where the number of variables exceeds the number of observations,
\item  and it should be shift and rotation equivariant, meaning that if the data are shifted or rotated the estimated parameters should transform accordingly.
\item It should be computable for high-dimensional data.
%\item Its estimation performance should not be contingent on the distribution of the outliers.
%\item It should be insensitive to the rate of contamination.
\item It should accurately describe the multivariate pattern of the majority of the observations, even when the data is heavily contaminated by outliers.
\item It should have a high breakdown point; a measure an estimator's robustness to outliers in the data. 
\item It should be insensitive to the dimensionality of the data.

\end{inparaenum}
Criteria (1)-(3) are natural for any PCA method. Criteria (4)-(5) relate to robustness. Criterion (6) is related to both concerns. We find that state-of-the-art robust PCA algorithms have most of these properties, but that, surprisingly, many instances can be found where they fail to satisfy Criterion (4). 
In this paper, we introduce a robust PCA algorithm, FastHCS, to meet these criteria (HCS for high-dimensional congruent subset).
In the next section we outline FastHCS. 
Then, in Sections~\eqref{hcs:s3} and~\eqref{hcs:s4} we compare it 
to several state-of-the-art methods on simulated data and three real data applications which show that in many settings only FastHCS can be relied upon to provide a robust PCA solution.

\section{FastHCS}\label{sfhcs}

Given an $n \times p$ data matrix $\pmb Y=\{\pmb y_i\}_{i=1}^n$ and for a fixed $2 \leqslant q<\min(p,n)$, the FastHCS algorithm searches for a subset of size at least  $h=\lceil(n+q+1)/2\rceil$ free of outliers
%of observations belonging to $\pmb X$ 
(this is the minimal value of $h$ such that there are at least $(q+1)$ clean observations in each candidate subset).

If $p>n$, FastHCS computes the mean-centered data matrix $\tilde{\pmb X}=\pmb{Y}-\pmb 1_n^\top(\ave_{i=1}^n\pmb y_i)$, and performs the kernel eigenvalue decomposition of $\tilde{\pmb X}\tilde{\pmb X}^\top=\pmb U\pmb L\pmb U^\top$ where $\pmb U$ is an $n\times r$ matrix, $\pmb L$ is an $r\times r$ matrix,  $r:=\rank(\pmb X)$ and $\pmb{L}$ is a diagonal matrix with the eigenvalues on the diagonal. 
Then, FastHCS works with the $n \times r$ matrix $\pmb X=\tilde{\pmb X}\tilde{\pmb X}^\top\pmb U (\pmb L)^{-1/2}$. The transformation from $\pmb{Y}$ to $\pmb X$ causes no loss of information or robustness since we retain all of the components corresponding to non-zero eigenvalues. However, this transformation reduces the computational cost of the subsequent steps of the algorithm. 
At the end of the algorithm, FastHCS reverses these transformations so that the returned parameter estimates are consistent with conventional PCA. When $n\leqslant p$, we simply set $\pmb X=\pmb Y$.

\subsection{The $I$-index $h$-subset}
\label{sec:Iindexsubset}
% \subsubsection{Candidate $h$-subsets}\label{sec:growsubsets} 
The $I$-index is a subset selection criterion first introduced in~\cite{hcs:VS13}  where it is used  to identify an outlier free subset to serve as the basis of the robust PCS location and scatter estimator. The $I$-index was designed to be  insensitive to the configuration of the outliers and consequently, as we show in that article, the fit found by FastPCS is nearly unaffected by the presence of outliers in the data (we refer to this as quantitative robustness). 
In~\citep{SV14} we further show that the PCS estimates also have the maximum possible breakdown point (we refer to this as qualitative robustness).
Robust location and scatter estimation are also important for PCA. In the PCS context, the $I$-index is applied to the observations in their original dimensionality, and one approach to achieving robust PCA would be to use the robust PCS covariance estimate as a starting point for PCA. However, this approach does not satisfy Criterion (3) for robust PCA since it is not possible to perform PCS when the number of dimensions is greater than the number of observations. This subsection describes how the $I$-index can be extended to the PCA context by applying it to projections of the data on to subspaces.

To begin, FastHCS draws $M$ random subsets of size $(q+1)$ from $\pmb X$ without replacement, where $M$ is given by:
      \begin{equation}\label{hcs:Mval}
        M=\left \lceil\frac{\log(0.01)}{\log(1-(e/n)^{q+1})}\right \rceil\;,
      \end{equation} 
and where $h\leqslant e<n$ is an integer specifying the number of uncontaminated observations,
so that the probability of getting at least one uncontaminated starting subset is at least 99\%~\citep{hcs:S81}. 
By default we set $e=h$. However, if the user is sure that the contamination rate of the sample is lower than $\slfrac{(n-h)}{n}$, we offer the possibility (as in~\cite{MY95}) 
of using this information to reduce the computational cost of 
running FastHCS. 
Denote these $(q+1)$-subsets as $\{H^m_0\}_{m=1}^{M}$.
The SVD decomposition of the observations indexed by $H^m_0$ is: 
\begin{equation}
\displaystyle\SVD_{i\in H^m_0}\left(\slfrac{(\pmb{x}_i-\pmb t^m_0)}{\sqrt{q}}\right)=\pmb{U}^m_0(\pmb L^m_0)^{1/2}(\pmb{P}^m_0), \nonumber
\end{equation}
where $\pmb t_0^m=\ave_{i\in H^m_0}\pmb{x}_i$ is the estimated center, $\pmb{L}^m_0$ is a diagonal matrix for which the non-zero elements $(\pmb{L}^m_0)_j\; j = 1, \dots, q$ are the descending eigenvalues of the PCA model fited to $\{\pmb x_i: i\in H^m_0\}$, and the eigenvectors $\pmb P^m_{0,q}$ are the first $q$ loadings of this model.
 Next, we compute the score matrix $\pmb S^m_0$ with $n$ rows $\pmb s^m_{0,i}$:
\begin{equation}
\pmb{s}^m_{0,i}=(\pmb{x}_i-\pmb t^m_0)\pmb{P}^m_{0,q},\quad1\leqslant i\leqslant n\nonumber
\end{equation}
which is the projection of the re-centered rows of $\pmb{X}$ on to
the subspace spanned by the first $q$ loadings of  $\{\pmb x_i:i\in H^m_0\}$. 
To measure the outlyingness of an $\pmb{s}_{0,i}^{m}$ to the members of $\{\pmb{s}^m_{0,i}:i\in H^m_0\}$, we will use its squared orthogonal distance to $\pmb a^{m}_k$, the direction normal to the 
hyperplane through $q$ members of  $\{\pmb{s}^m_{0,i}:i\in H^m_0\}$ drawn at random:
\begin{equation}\label{hcs:basedist}
    d^2_i(\pmb a^{m}_k,\pmb S^m_0)=\slfrac{\left((\pmb{s}_{0,i}^{m})^\top\pmb a^{m}_k-1\right)^2}{||\pmb a^{m}_k||^2}\;, \nonumber
\end{equation} 
and, to remove the dependence of this measure on the direction $\pmb a_k^m$, we average it over $K$ such directions $\{\pmb a^{m}_k\}_{k=1}^K$: 
\begin{equation}\label{hcs:growdist}
 D_i(H^m_{0})=\displaystyle\ave_{k=1}^{K}  
    \frac{ d^2_i(\pmb a^{m}_k,\pmb S^{m}_0)}
         {\displaystyle\ave_{i\in H^{m}_{0}} d^2_i(\pmb a^{m}_k,\pmb S^{m}_0)}
,\;\;1\leqslant i\leqslant n.
\end{equation}
(In Remark 1 below we discuss how we set the value of the parameter $K$).
The denominator in Equation~\eqref{hcs:growdist} normalizes these distances across the directions $\pmb a_k^m$. 

We can now describe the computation of the first step of FastHCS. 
For a given a $(q+1)$-subset  $H^m_0$ of $\{1:n\}$  and its corresponding matrix $\pmb S^m_0$, 
Algorithm 1 returns and $h$-subset $H^m$ of 
indexes of $\{1:n\}$ using an iterative procedure we call \emph{growing steps}. 
In each step $w$, $H^m_w$ is updated and contains the indexes of the  $\omega_w$ observations with smallest values of $D_i(H^m_{w-1})$. The value of $\omega_w$ itself increases incrementally from 
$\lceil\slfrac{(n-q-1)}{(2W)}\rceil+q+1$ to 
$h$ in $W$ steps. 
These steps do not have 
a convergence criterion, so the number 
of iterations $W$ must be set in advance (In Remark 1 below we discuss how we set the value of the parameter $W$). 

\vskip0.15cm
\hrule
\vskip0.15cm
\noindent \begin{equation}
\text{Algorithm 1: growingStep$(H_0^m,\pmb X,q)$} \nonumber
\end{equation}
\hrule
\begin{tabbing}

for \=$w=1$ to $W$ do:\\
\indent \>
\begin{math}
  D_i(H^m_{w-1}) \gets \displaystyle\ave_{k=1}^{K}  
    \frac{ d^2_i(\pmb a^{m}_k,\pmb S^{m}_{0})}
         {\displaystyle\ave_{i\in H^{m}_{w-1}} d^2_i(\pmb a^{m}_k,\pmb S^{m}_{0})}
,\;\;1\leqslant i\leqslant n
\end{math}\\
\indent\>set 
\begin{math}
  \omega_w\gets\lceil\slfrac{(n-q-1)w}{(2W)}\rceil+q+1            
\end{math}\\
\indent \>set 
\begin{math}
  H^m_w \gets \left\{i: D_i(H^m_{w-1}) \leqslant D_{(\omega_w)}(H^m_{w-1})\right\}
\end{math}\\
%\>\>\>     \text{ (`concentration step')} \\
end for\\
\begin{math}
H^m \gets H^m_W
\end{math}
\end{tabbing}
\hrule
\vskip0.3cm

 After growing $M$ candidate $H^m$'s, FastHCS evaluates each using a criterion we call the $I$-index, and fits a robust PCA model to the $H^m$ having smallest value of the $I$-index. For a given $h$-subset $H^m$ and direction $\pmb a^m_k$, we define a subset $H^{m}_k$ that is optimal with respect to $\pmb a^m_k$ in the sense that it indexes 
 the $h$ observations with the smallest values of $d^2_i(\pmb a^{m}_k,\pmb S^{m}_0)$. 
 More precisely, denoting $d_{(h)}$ the $h^{th}$ order 
statistic of a vector $\pmb d$, we have:
\begin{equation}
H^{m}_k=\{i:d_{i}^2(\pmb a^{m}_k,\pmb S^{m}_0)\leqslant 
d^2_{(h)}(\pmb a^{m}_k,\pmb S^{m}_0)\}. \nonumber
\end{equation}
Then, we define the $I$-index of an $H^m$ along $\pmb a^{m}_k$ as
\begin{equation}\label{mcs:crit1}
  I(H^m,\pmb{S}^m_0,\pmb a^{m}_k)=
    \log\left(\frac{\displaystyle\ave_{i\in H^m} d^2_i(\pmb a^{m}_k,\pmb S^{m}_0)}{\displaystyle\ave_{i\in H^{m}_k} d^2_i(\pmb a^{m}_k,\pmb S^{m}_0)}\right),
\end{equation}%%%
with the convention that $\log(0/0):=0$. The measure 
$I(H^m,\pmb{S}^m_0,\pmb a^{m}_k)$ is always positive
 and increases the fewer members $H^m$ shares with $H^{m}_k$ along the direction $\pmb a^{m}_k$. This is because, 
 for a given direction $\pmb a^{m}_k$, the members of 
 $H^{m}_k$ not in $H^m$ will decrease the denominator 
 in Equation~\eqref{mcs:crit1} without affecting
the numerator, increasing the overall ratio. 
 As in the growing steps, we 
 remove the dependence of Equation
 \eqref{mcs:crit1} on the directions 
 $\pmb a^{m}_k$ by
considering the average over $K$ directions:
\begin{equation}\label{mcs:crit2}
      I(H^m,\pmb{S}^m_0)=\displaystyle\ave_{k=1}^{K} I(H^m,\pmb{S}^m_0,\pmb a^{m}_k)\;.
\end{equation}
Finally, FastHCS selects as $H^I$ the candidate 
$h$-subset $H^m$ with lowest $I$-index.

Given $H^{I}$, we denote the PCA parameters 
 corresponding to $H^{I}$ as  
$(\pmb t^{I}, \pmb L_q^{I},\pmb P_q^{I})$ and obtain 
them as follows:
\begin{eqnarray}
\displaystyle\SVD_{i\in H^{I}}\left(\slfrac{(\pmb{y}_i-\pmb t^{I})}{\sqrt{h-1}}\right)=\pmb{U}^{I}(\pmb L^{I})^{1/2}(\pmb{P}^{I})^\top,\nonumber
\label{eq:rawFit}
\end{eqnarray}
where $\pmb t^{I}=\ave_{i \in H^{I}}\pmb{y}_i$.
FastHCS computes these parameters on the full space of the data set, $\pmb Y$, rather than on the space of $\pmb S_0^I$, to increase their accuracy.
Algorithm 2 summarizes the I-index step of FastHCS. 

\vskip0.15cm
\hrule
\vskip0.15cm
\noindent \begin{equation}
\text{Algorithm 2: IStep$(\pmb X,q)$} \nonumber
\end{equation}
\hrule
\begin{tabbing}

\indent \=for \= $m=1$ to $M$ do:\\ 
\indent \>\> 
\begin{math}
  H^m_0\gets\{\text{random } $(q+1)$-\text{subset from } 1:n\}    
\end{math} \\
\indent \>\> 
\begin{math}
H^m \gets \text{growingStep}(H^m_0,\pmb X,q)
\end{math} \\
\indent  \>\>
\begin{math}
  I(H^m,\pmb S^m_0) \gets \displaystyle\ave_{k=1}^{K} I(H^m,\pmb S^m_0,\pmb a_{k}^m)
\end{math}
\\
\>end for\\
\indent \> $H^{I}\gets\displaystyle\underset{H^1,\ldots,H^M}{\argmin}\;I(H^m,\pmb S^m_0)$ \\
 \indent \> return $(\;\pmb t^{I},\; \pmb L^{I}_{q},\pmb P^{I}_{q})$
\end{tabbing}
\hrule
\vskip0.3cm

\begin{remark}
Through experiments, we find that increasing $K$ above 25 or $W$ above 5 does not noticeably improve the performance of the algorithm (though it increases its computational cost), so we set these parameters to those values.
Those experiments were carried on the outlier configurations discussed in Sections 3 and 4 as well as additional configurations enabled by the simulation suite provided with the Online Resources (Section 4). Because such experiments cannot cover all possible configurations of outliers,  
we focused on those configurations singled out as most challenging in the literature on robust PCA.  
\end{remark}

\begin{remark}
\textbf{Exact fit:} When the $h$ members 
of a subset $H'$ lie on a subspace $\pmb \Pi_r\in \mathbb{R}^r$ with $1<r\leqslant q$ the 
numerator and denominator of $I(H',\pmb{S}'_0,\pmb a'_k)$ will be the same for any direction $\pmb a'_k$ through members of $H'$~\citep{SV14} so that $I(H',\pmb{S}'_0)=0$. Then, $H^I=H'$ and $\pmb P^{I}_r$ will correspond with $\pmb \Pi_r$. 
 In such situations, FastHCS will return the index of the members of $\{i:((\pmb x_i-\pmb t^{I})\pmb P_r^*)^2=0\}$. This behavior is called exact fit~\citep{mcs:MMY06}.
\end{remark}

\begin{remark}
\textbf{Breakdown point:}
The (finite sample) breakdown point of an estimator referred to in Criterion (5) is the smallest proportion of observations
that need to be replaced by arbitrary value to drive the estimates to arbitrary values~\citep{hcs:D82}.
Naturally, a higher breakdown point is better, and the maximal breakdown point achievable 
in the PCA context is essentially fifty percent.

Both the growing step and the I-index use distances 
computed on subspaces to derive a measure of outlyingness. Since they are 
restricted to this view of the data, they are vulnerable to outliers that 
appear consistent with the majority on a subspace, but are outlying with 
respect to it (Appendix A details the specific configurations of outliers 
giving rise to this issue). Therefore, fits based on $H^{I}$ alone will 
not have maximum breakdown and the procedure presented above must be 
 combined with a second, computational expedient, ancillary procedure 
 to ensure that the final FastHCS estimates do. 
\end{remark}

\subsection{The Projection Pursuit $h$-subset}
\label{sec:PP}
Although experiments, such as those in Sections~\ref{hcs:s3} and~\ref{hcs:s4}, show that the $I$-index rarely selects contaminated subsets, it is vulnerable to specific configurations of outliers (see Appendix A).
To guard against these, FastHCS uses a robust projection-pursuit (PP) approach to identify a second subset of observations, $H^{PP}$. The PP approach proceeds by assigning an outlyingness score to each observation:
\begin{equation}
d^{PP}_i(\pmb Y) = \underset{\pmb v \in B}{\max} \frac{\left | \pmb y_i \pmb v - \med(\pmb y_j \pmb v)\right |}{\MAD(\pmb y_j \pmb v)}
\end{equation}
where $B$ contains 1000 directions $\pmb v$, each given by two data points drawn randomly from the sample, $\med(\pmb y_j \pmb v)$ is the median of $\{\pmb y_j \pmb v,  j = 1,\dots, n\}$ and $\MAD(\pmb y_j \pmb v) = \med|\pmb y_j \pmb v - \med(\pmb y_l \pmb v)|$. The PP method is orthogonaly invariant and computationally expedient. A version of the PP algorithm
 is used as an initial step in ROBPCA~\citep{hcs:HRV05,hcs:DH09}, a popular robust PCA algorithm.

\subsection{Selecting the final PCA model}
\label{sec:selectionCrit}
Consider the subset $H^\bullet:= H^{I} \cap H^{PP}$. 
Because $h\geqslant[n/2]+1$, it holds that
$|H^\bullet|\geqslant q$ and $H^\bullet$ is free of outliers whenever either one of $H^{I}$ or $H^{PP}$ is. 
We propose to exploit this fact to select between the I-index and PP-based models. Denote $H^-=H^{PP} \setminus H^{I}$ and 
\begin{eqnarray}
D(\pmb Y,H^I,H^{PP})=\ave_{j=1}^q\log\frac{\ave_{i \in H^I}((\pmb y_i-\pmb t^{I})\pmb P^{I}_j)^2}{\var_{i \in H^\bullet}(\pmb y_i\pmb P^{I}_j)}\nonumber\\
-\max_{j=1}^q\log\frac{\ave_{i \in H^\bullet}((\pmb y_i-\pmb t^{PP})\pmb P^{PP}_j)^2}{\var_{i \in H^-}(\pmb y_i\pmb P^{PP}_j)},
\label{eq:criterion}
\end{eqnarray}
with the assumption that $\log(0/0)=0$. When $D(\pmb Y,H^I,H^{PP})>0$ (or if $\displaystyle\max_{j=1}^q\var_{i \in H^{-}}(\pmb y^\varepsilon_i\pmb P^{PP}_j)=0$) the final FastHCS parameters $(\pmb t^{*},\pmb L_q^{*},\pmb P_q^{*})$ will be equal to $(\pmb t^{PP},\pmb L_q^{PP},\pmb P_q^{PP})$ and the final FastHCS subset $H^*$ is set as $H^{PP}$.
 As we show in Appendix~\ref{app:AB}, this selection rule ensures that the FastHCS fit has a high breakdown point. Our approach is similar to the ROBPCA algorithm which also selects from among two candidate subsets in the final stage of the algorithm. ROBPCA selects the subset whose eigenvalues have the smallest product. However, depending on the configuration of the outliers and the rate of contamination, it is possible for a contaminated subset to have smaller eigenvalues than an uncontaminated one~\citep{SV14}, and so to end up being selected by the criterion used in ROBPCA. In contrast, the selection criterion we propose controls (through the denominators in Equation~\eqref{eq:criterion}) for the relative scatter of the two subsets and so it is not biased towards subsets containing many concentrated outliers. Naturally, criterion \eqref{eq:criterion} is designed to favor the $I$-index based model whenever doing so does not cause breakdown.

\subsection{Outlyingness to the PCA model}
Two concepts of distance are used to assess the outlyingness of an observation with respect to a PCA model, and cut-off values for both of these can be used to classify outliers~\citep{hcs:HRV05}. 
The first is the orthogonal distance ($\od$) of the observation to the PCA model space:
\begin{eqnarray}\label{hcs:odist}
{\od}_i(\pmb t,\pmb P_q)=||\pmb x_i-\pmb t-(\pmb x_i-\pmb t)\pmb P_q(\pmb P_q)^\top||
\end{eqnarray}
Assuming multivariate normality of the observations on which the PCA model is fitted, a cut-off can be obtained for the $\od$ statistics using the Wilson-Hilferty transformation of the 
$\od$s into approximately normally 
distributed random variables: 
\begin{eqnarray}\label{hcs:cut-off}
c_e(\pmb t, \pmb P_q, H)&=&\left(\ave_{i\in H}{\od}^{2/3}_i(\pmb t,\pmb P_q)\right.\nonumber\\
&&\left.+\Phi^{-1}_{0.975}\sqrt{\frac{\var_{i\in H} {\od}^{2/3}_i(\pmb t,\pmb P_q)}{\chi^2_{e/n,1}}}\right)^{3/2}
\end{eqnarray}
where $\chi^2_{e/n,1}$ is the $e/n$ quantile 
of the $\chi^2$ distribution with one degree of freedom, and $H$ indexes a subset of observations.
The second measure of outlyingness is the score distance ($\sd$) of the observation on the PCA model space: 
\begin{eqnarray}\label{hcs:sd}
{\sd}_i(\pmb t,\pmb L_q,\pmb P_q)&=&\sqrt{\left((\pmb x_i-\pmb t)\pmb P_q\right)\pmb (\pmb L_q)^{-1}}\nonumber\\
&&\;\;\;\;\;\;\;\;\;\;\;\;\;\;\;\;\;\;\;\;\overline{\left((\pmb x_i-\pmb t)\pmb P_q\right)^\top}. 
\end{eqnarray}
 A 97.5\% cut-off for the $\sd$ statistics can be obtained using a $\sqrt{\chi^2_{0.975,q}}$ distribution.
 
 In inferential applications, PCA theory typically assumes multivariate normality, though ellipticity is sufficient for many of the hypotheses of interest to PCA-based inference, see~\citep{hcs:J86} and~\citep[pages 49,55,394]{Jolliffe2002}. In any case, robust PCA performs inference with a model fitted on the non-outlying observations, so the distributional assumption pertains to only this subset of the data. Conversely, no assumptions are made on the distribution(s) of the outliers.
 
\subsection{Computational considerations}
The computational complexity of FastHCS is determined by the $I$-index and PP subset selection components. The complexity of the PP-based approach is $\mathcal{O}(qnp)$. This is dominated by the time complexity of the $I$-index, which scales as $\mathcal{O}(q^3+nq)$ for each starting $(q+1)$-subset. 
Except when $q$ and $n$ are small (smaller than about 5 and 2000 in our tests) FastHCS is not
the quickest of the robust PCA methods we considered (in general, we find that PcaL is).
The 'Fast' qualification in this context ("FastHCS") is used to distinguish
 the algorithm based on random sub-sampling from the na\"{i}ve one based
 on exhaustive enumeration of all possible starting points, the latter being usually not computable in practice. 
The computing time of FastHCS grows similarly in $n$ to other methods we discuss in this paper, while it is the most sensitive to increases in $q$, with computation times being comparable until around $q=12$. For higher $q$, FastHCS is the slowest overall.
In practice, FastHCS becomes impractical for values of $q$ much larger than 25.  
Nevertheless, the overall time complexity of
FastHCS grows with $q$, instead of $p$,
 making it a suitable candidate for  
high-dimensional applications, and satisfying Criterion (3) for a robust PCA method.
Moreover, FastHCS belongs to the class of 
so called `embarrassingly parallel' algorithms, i.e. 
 its time complexity scales as the inverse
 of the number of processors, meaning 
 it is well suited to benefit
 from modern computing environments. 
 To enhance user experience, we implemented FastHCS in 
\cpp $\,$ and wrapped it in an portable, open source
\texttt{R} package~\citep{rcore} distributed through
\texttt{CRAN} (package \texttt{FastHCS})

\section{Simulation Study}\label{hcs:s3}

In this section we evaluate the behavior of FastHCS against three other robust PCA methods: the ROBPCA \citep{hcs:HRV05}, PcaPP~\citep{hcs:cr05} 
and PcaL~\citep{hcs:L99} algorithms. Although other methods for high-dimensional outlier detection exist, these are particularly comparable with FastHCS: all three are PCA algorithms, satisfying Criteria (1)-(3) of a robust PCA method.
ROBPCA, PcaPP and PcaL were 
computed using the \texttt{R}~\citep{rcore} package 
\texttt{rrcov}~\citep{mcs:TF09} with default settings 
except for the robustness parameter 
\texttt{alpha} for ROBPCA which we 
set to 0.5, the value yielding maximum 
robustness and the value of $\texttt{k}$ 
 which we set to $q$ for all the algorithms.
   Our evaluation criterion is the empirical
bias, a quantitative measure of robustness of
  a fit. In Appendix C we explain the motivation
 for this choice (in the Online Resources 
 we also report the results obtained using 
 alternative evaluation criteria).
  
\subsection{Empirical bias}
Given an elliptical distribution 
$\mathscr{E}_p$ with location vector $\pmb \mu^u$ and covariance matrix $\pmb \varSigma^u$ (the superscript $u$
stands for uncontaminated)
 and an arbitrary distribution   
$\mathscr{F}^c$ (the superscript $c$
stands for contaminated), consider 
the $\varepsilon$-contaminated model
\begin{eqnarray}
\mathscr{F}^{\varepsilon}=(1-\varepsilon)\mathscr{E}_p(\pmb\mu^u,\pmb\varSigma^u)+\varepsilon\mathscr{F}^c.\nonumber
\end{eqnarray}
For a fixed $q<p$ denote $\pmb\varSigma^u_q$ the rank $q$ 
approximation of $\pmb\varSigma^u$ and $\pmb V_{q}^{}=\pmb P_q^{}\pmb L_q^{}\pmb P_q^\top$ an estimator of $\pmb\varSigma^u_q$.
The (empirical) bias measures the difference between $\pmb V_{q}$ and $\pmb \varSigma^u_q$. For this, we will consider 
more specifically the shape component of this difference which is called the shape bias. 
Given these two (rank reduced) covariance matrices, recall that the corresponding shape matrices are defined by $\pmb \Gamma^u=|\pmb\varSigma^u|^{-1/q}\pmb\varSigma^u_q$ and $\pmb G_q=|\pmb V_q|^{-1/q}\pmb V_q$.
For an estimator of $\pmb V_q$, all the information about the shape bias is 
contained in the matrix $(\pmb \Gamma^u)^{-1/2}\pmb G_{q}(\pmb \Gamma^u)^{-1/2}$, or equivalently its condition 
number~\citep{mcs:YM90}:
\begin{eqnarray*}
\mbox{bias}(\pmb V_{q})=
\log\frac{\lambda_1\left((\pmb \Gamma^u)^{-1/2}\pmb G_{q}(\pmb \Gamma^u)^{-1/2}\right)}{
\lambda_q\left((\pmb \Gamma^u)^{-1/2}\pmb G_{q}(\pmb \Gamma^u)^{-1/2}\right)}\;, \nonumber
\end{eqnarray*}
where $\lambda_1$ ($\lambda_q$) is the largest 
($q^{th}$) eigenvalue of the positive-semidefinite matrix $(\pmb \Gamma^u)^{-1/2}\pmb G_{q}(\pmb \Gamma^u)^{-1/2}$.
 Evaluating the maximum bias of $\pmb V_{q}$ is an empirical
 matter: for a given sample, it 
  depends on the dimensionality of the data, 
  the rate of contamination by outliers, the
  distance separating them from the uncontaminated  
 observations, and the spatial 
configuration of the outliers ($\mathscr{F}^c$). 
However, because all the algorithms we compare are rotation and 
shift equivariant, w.l.o.g. we can focus on configurations 
where $\pmb\varSigma^u$ is diagonal, and $\pmb \mu^u=\pmb 0_p$ (a $p$-vector of zeros),
 greatly reducing the number of scenarios we need to consider.
Since the effect of contamination
 is presumably most harmful when the outlier belongs
to the subspace spanned by $\pmb \Pi_q^{u\perp}$ (the orthogonal 
complement of $\pmb \Pi^u_q$) we, concentrate on
the class of outlier configurations satisfying these conditions~\citep{hcs:M05}.
In the simulation 
 results shown in Section~\ref{mcs:s5}, the outliers belong 
  to the subspace spanned by the eigenvector corresponding to
  the $(q+1)$-th eigenvalue of $\pmb\varSigma^u$ (as in \cite{hcs:HRV05}) whereas the simulation 
 settings shown in the Online Resources the outliers belong to the subspace spanned by all the components of $\pmb \Pi_q^{u\perp}$
  (as is done in~\cite{hcs:M05}).

\subsection{Outlier configurations}

To quantify the robustness of the four algorithms, 
we generate many contaminated data sets
 $\pmb X^\varepsilon$ of size $n$ with 
$\pmb X^{\varepsilon}=\pmb X^u\cup\pmb X^c$ where 
$\pmb X^u$ and $\pmb X^c$ are 
the uncontaminated and contaminated part of 
the sample. 
In all simulations, the center of the uncontaminated data $\pmb \mu^u= \pmb 0_{1 \times p}$ its $\pmb\varSigma^u$ is either $\pmb\varSigma^u$ or $10^{-4}\pmb\varSigma^u$. We show results where $p\in\{100,400\}$, $q\in\{5,10,15\}$, $n=200$, and $\varepsilon$ is one of $\{0.1,0.2,0.3,0.4\}$
%Already said above now:::To parameterize $\pmb X^u$ we will set $\pmb \mu^u = \pmb 0_p$. 

To facilitate comparison, we consider a generalization to arbitrary values of $q$ of the parametrization of $\pmb \varSigma^u$ used in~\citep{hcs:HRV05}.
% \item $\pmb L_M$: the values of the diagonal elements slowly decrease from the first to the $q^{th}$, and then drop sharply. More precisely, the values of the diagonal elements of 
% $\pmb\varSigma^u$ are $20(1+(1-j+q)/2),\;j=1,\dots,q$ and $p^{-1}(p-j+1)+1,\;j=q+1,\dots,p$. This is a generalization to arbitrary values of $q$ of the parametrization used in~\citep{hcs:M05}.
To define this matrix, $\pmb L$, we set the values of the first $q$ elements of the diagonal of $\pmb\varSigma^u$ so that they decrease exponentially and do not drop abruptly before the remaining, smaller, entries.
More precisely, the first $q$ entries of the diagonal of  $\pmb \varSigma^u$ are the first $q$ Fibonacci numbers and the entries $q+1,\dots,p$ are linearly decreasing as ($0.1, \dots, 0.001$).  
In Section 2 of the Online Resources, we also provide results using a covariance matrix proprosed by~\citep{hcs:M05}. 

Our measure of robustness, the bias, 
depends on the distance between the 
outliers and the non outlying observations which 
we will measure by
\begin{equation}\label{mcs:nu}
\nu = \min_{i\in I^c} \sqrt{(\pmb{x}_i^\top(\pmb\varSigma^u)^{-1}\pmb{x}_i^{})/\chi^2_{0.975,p}}, \,\,
\end{equation}
 where $I^c$ is an indicator for the observations coming from $\pmb X^c$. (A more detailed description of how we set the location of the outliers can be found in Section 3 of the Online Resources.)
 In the simulations, the distance $\nu$ separating the outliers from the good data is one of $\{1,\ldots,10\}$

We consider two outlier configurations frequently used in the robust PCA literature~\citep{hcs:HRV05, hcs:M05}:
% because they are expected to be challenging
\begin{inparaenum}[(a)]
 \item Shift outliers: $\pmb\varSigma^c=\pmb\varSigma^u$ and $\pmb\mu^c$ chosen to satisfy Equation~\eqref{mcs:nu};
\item Point-mass outliers: all the outliers are concentrated around a single point at a distance $\nu$ from $\pmb X^u$. To obtain this effect, we set $\pmb\varSigma^c = 10^{-4}\pmb\varSigma^u$.
\end{inparaenum}
Both of these outlier configurations are relevant in practical applications where they are, for example, similar to certain types of sensor faults and contamination scenarios.

For FastHCS, the number of initial $(q+1)$-subsets $M$ is given as in Equation~\eqref{hcs:Mval}, with $e/n=0.6$. The \texttt{rrcov} implementations for ROBPCA and PcaPP include hardcoded values for the number of starting subsets presumed by their authors to be sufficient for these methods. PcaL does not require starting subsets.
Section 4 of the Online Resources explains how the reader can use code we supply to replicate all results from this section.

In Figures \ref{hcs:sim1} to \ref{hcs:sim2}, we display
 the bias curves as lattice plots \citep{mcs:D08} for 
 discrete combinations of $p$, $q$ and $\varepsilon$.  
 In all cases, we expect the outlier
 detection problem to become monotonically harder as
 we increase $q$ and $\varepsilon$, so little 
information will be lost by considering a discrete 
grid of a few values for these parameters.
The configurations also depend on the distance separating
 the data from the outliers. 
 Here, the effects of $\nu$ on the bias are harder to foresee: clearly 
  nearby outliers will be harder to detect but misclassifying   
distant outliers will increase the bias more. Therefore, we 
 will test the algorithms for many values (and chart 
the results as a function) of $\nu$. For each algorithm, a solid colored
 line will depict the median, and a dotted line (of the same 
color) the 75th percentile of $\mbox{bias}(\pmb V_q)$. 
Each figure is based on 12000 simulations. 

\subsection{Simulation results}\label{mcs:s5}

Figure~\ref{hcs:sim1} displays the bias curves corresponding to the fits found by the algorithms for $p=100$ for the shift (right) and point-mass (left) configurations. Regardless of the spatial configuration of the outliers or the value of $\varepsilon$, the fits found by PcaPP and PcaL generally have high values of $\mbox{bias}(\pmb V_q)$. As it turns out, PcaPP and PcaL will show poor performance on all of the remaining simulations as well. Since this poor performance is consistent, we will not discuss it in detail. The performance of ROBPCA is substantially better than the previous two algorithms, but it becomes increasingly unreliable as $q$ increases. FastHCS shows almost no bias. Furthermore, we see that in some cases even after the bias curves of ROBPCA have re-descended, they are still above those of FastHCS. Given that this gap increases with $\varepsilon$, we infer that the eigenvalue estimation of ROBPCA is still influenced by the outliers, even when the eigenvectors are correctly estimated. FastHCS estimates both correctly.

\begin{figure}[ht!]
\centering
\includegraphics[width=0.95\textwidth]{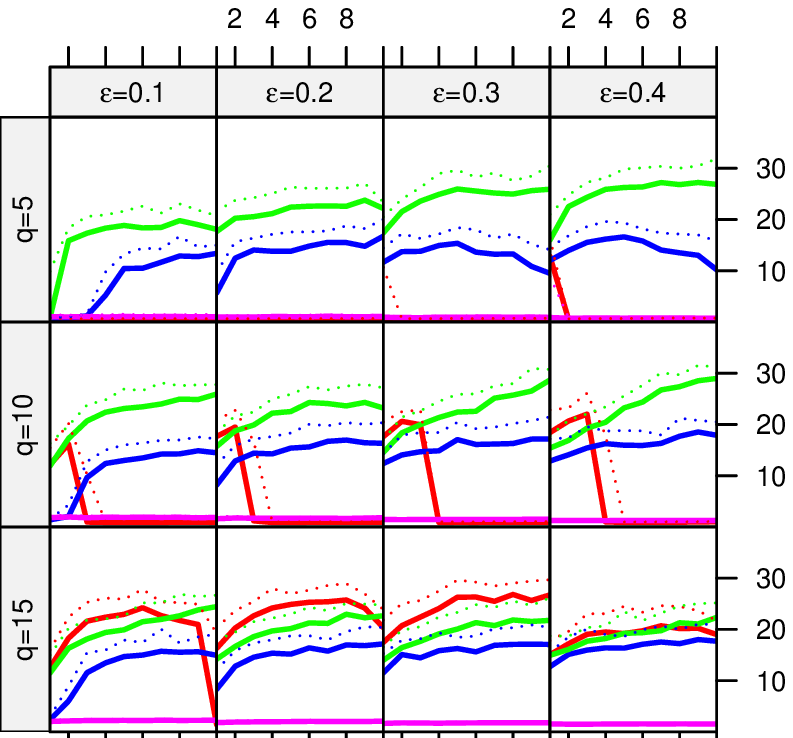}
\includegraphics[width=0.95\textwidth]{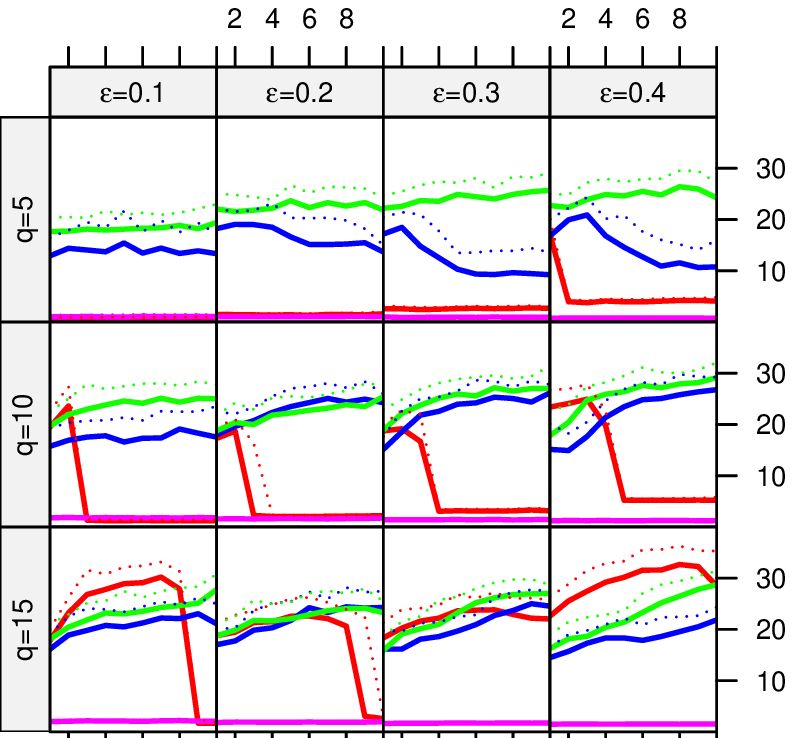}
\caption{$\mbox{bias}(\pmb V_q)$ for $p=100$, shift (top) and point-mass (bottom) as a function of $\nu$.
         \textcolor{ROBPCA}{\underline{ROBPCA}},
         \textcolor{PcaPP}{\underline{PcaPP}}, 
         \textcolor{PcaL}{\underline{PcaL}},
         \textcolor{HCS}{\underline{FastHCS}}.}
\label{hcs:sim1}
  \vspace{-1.0em}
\end{figure}
We next consider the high dimensional case of $p>n$. Figure~\ref{hcs:sim2}. Across all scenarios, the results are comparable to those in seen in Figure~\ref{hcs:sim1}.
FastHCS is the best performing method, being unaffected by the outliers, while the other methods show high biases on most settings.

\begin{figure}[ht!]
\centering
\includegraphics[width=0.95\textwidth]{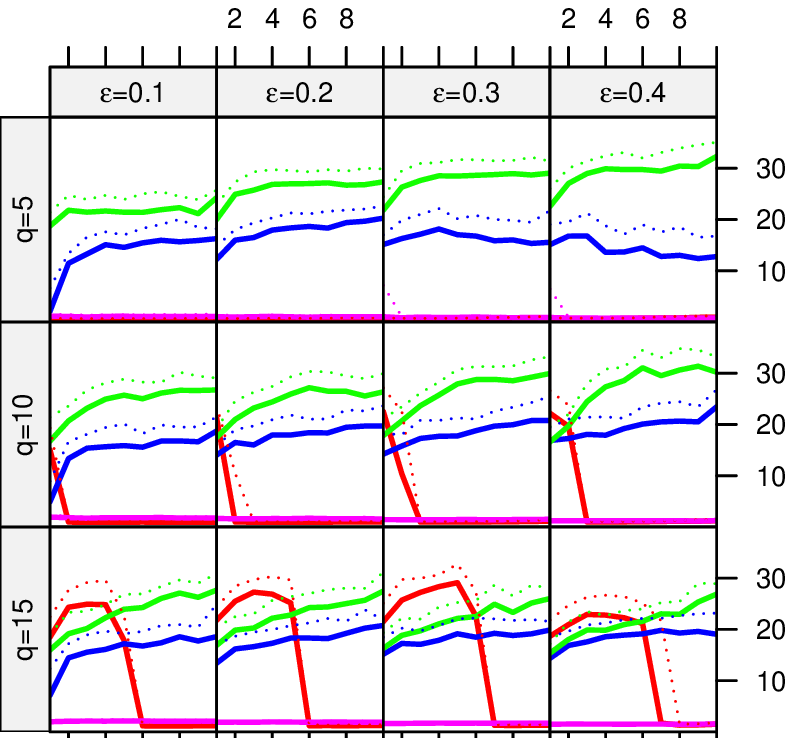}
\includegraphics[width=0.95\textwidth]{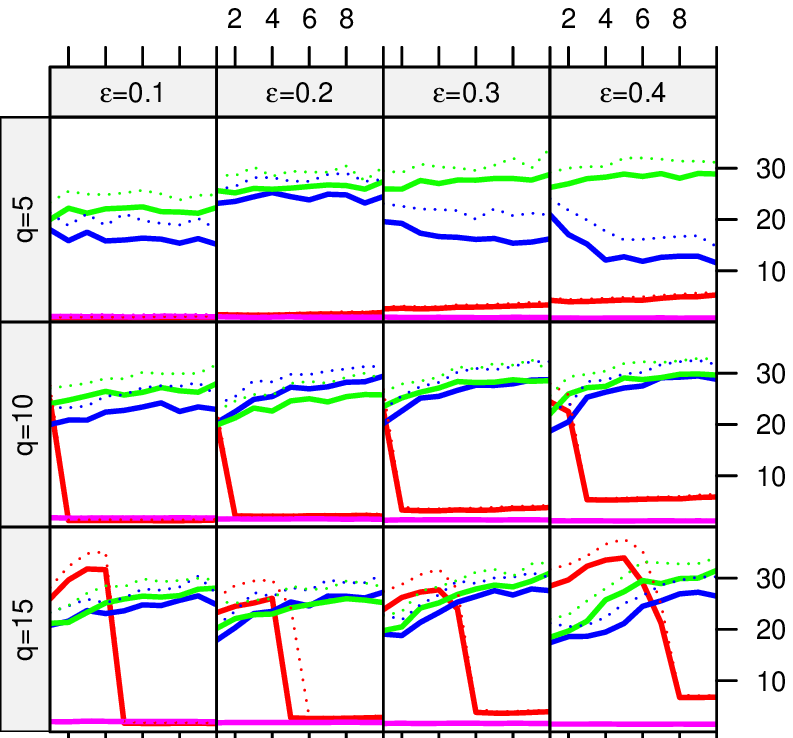}
\caption{$\mbox{bias}(\pmb V_q)$ for $p=400$, shift (top) and point-mass (bottom) as a function of $\nu$.
         \textcolor{ROBPCA}{\underline{ROBPCA}},
         \textcolor{PcaPP}{\underline{PcaPP}}, 
         \textcolor{PcaL}{\underline{PcaL}},
         \textcolor{HCS}{\underline{FastHCS}}.}
\label{hcs:sim2}
  \vspace{-1.0em}
\end{figure}
Over all of the scenarios, FastHCS shows almost no bias, despite challenging outlier configurations. Furthermore, the bias curves corresponding to the fits found by FastHCS are also less variable: throughout, the 75th percentile of the bias corresponding to the FastHCS fit is typically closer to the median bias than is the case for the other algorithms. These findings indicate that FastHCS meets Criteria (4)-(5) for a robust PCA method. In contrast, we find that the performance of the other methods vary with the configuration of the outliers, the rate of the contamination, and the dimensionality of the $q$-subspace. In Section 5 of the Online Resources, we re-analyze these simulation results, giving similar plots for a measure of outlier misclassification, as well as the principal angle measure and $\mbox{bias}(\pmb P_q)$, two measures of quality of fit focusing on the loadings.

An extended simulation study shows that the results we present above are robust the choice of different simulation settings (for example, those used in~\citep{hcs:M05}),
and different bias measurments criterions. However, since the outcome of the extended study is nearly identical to the one we present in this section, we have relegated these results in the Online Resources.

\section{Real data applications}\label{hcs:s4}

We next apply the algorithms to three real data examples. 
We selected these examples because in each the 
observations in the data can be separated into two subgroups from which we construct a majority and an outlier group.
They are taken from three fields that regularly use PCA: character recognition, chemometrics
and genetics. A feature shared by all of these data sets is that the variables within each are measured on the same scale. Data sets with this property were selected to remove the ancillary problem of data standardization. If the variables are not on the same scale, it is common practice in PCA modelling to standardize the data, but the choice of how to do so robustly adds a layer of subjectivity to the results. For the interested reader, the data sets used 
in this section are included in the \texttt{FastHCS} package. Section 6 of the Online Resources explains how the reader can use code we provide to replicate all results in this section.

The implementations of ROBPCA and PcaPP we use do not have an option to set the seed, but to ensure
reproducibility of the results for FastHCS, we set 
\texttt{seed=1}. PcaL is a deterministic algorithm and uses no seeds.
As in the simulations, we run each algorithm
 with default settings, except for the \texttt{alpha} parameter in ROBPCA which we set to 0.5.
 To illustrate the outlier detection capabilities of the algorithms, we display diagnostics plots. These show the $\od$ and $\sd$ values for each observation, divided by the cut-off values in Equations \eqref{hcs:cut-off} and~\eqref{hcs:sd} to put each of the methods on a comparable scale. 

\subsection{Selecting the number of components}
\label{sec:selectq}
We recommend using as large a value of $q$ as possible for FastHCS, since this improves its outlier detection performance. However, to avoid the curse of dimensionality, it is also advised to set $q<n/5$~\citep{hcs:HRV05}.
In all the examples that follow, we select a relatively high number of components, $q=15$, to strike a balance between computation time and accuracy. Once the outliers have been detected, components with large eigenvales can be analyzed and used to construct a PCA model of the good data.
One may also wish to use a selection criterion, such as the scree chart or contribution to variance. In Section 7 of the Online Resources, we also show results using the latter of these approaches in an extended analysis. In that analysis, the chemometrics data set illustrates how robust PCA methods parametrized based on a parsimonious, eigenvalue-based criterion, may miss outliers on minor components, even when the majority of the data may be modelled using a parsimonious model.

\begin{figure*}
\centering
\floatbox[{\capbeside\thisfloatsetup{capbesideposition={right,top},capbesidewidth=4cm}}]{figure}[\FBwidth]
{\caption{The 350 vectors of Fourier coefficients 
of the character shapes for the Multiple Features data set. The first 150 curves 
(corresponding to observations with labels '0') are show in 
the top panel in light orange. The Main group (200 curves) corresponding to observations with labels '1' are shown in the bottom panel.}\label{hcs:rd1}}
{\includegraphics[width=0.65\textwidth]{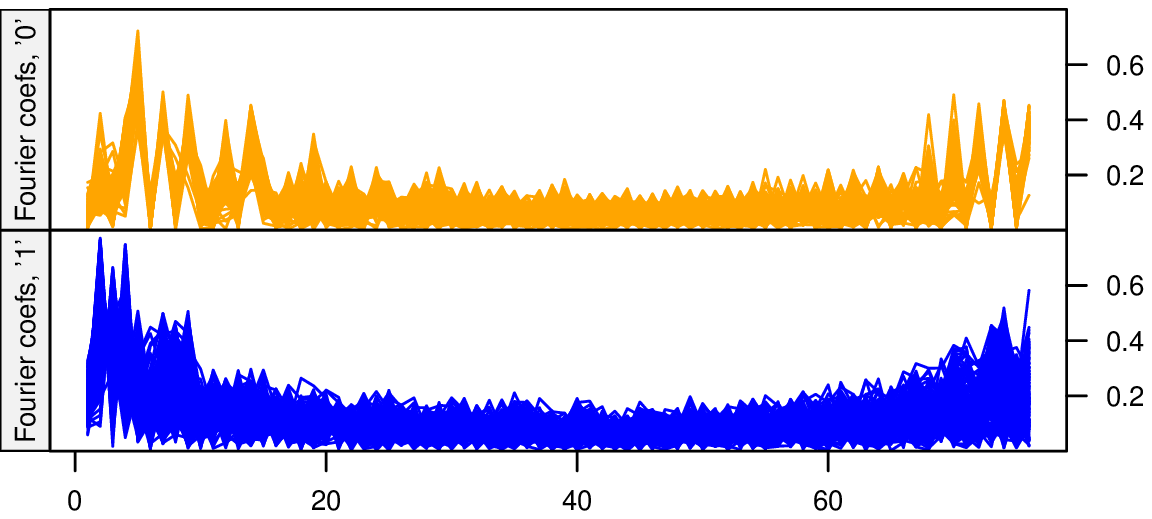}}
\vspace{-1.0em}
\end{figure*} 

\begin{figure*}
\centering
\floatbox[{\capbeside\thisfloatsetup{capbesideposition={right,top},capbesidewidth=4cm}}]{figure}[\FBwidth]
{\caption{Diagnostic plots of the scaled score and orthogonal distances of the Fourier coefficients of the numerals for the robust PCA fits corresponding to the four algortihms for the Multiple feature data set. The observations corresponding to numerals with labels '1' ('0') are shown as dark blue circles (light orange triangles).}\label{hcs:rd2}}
{\includegraphics[width=0.65\textwidth]{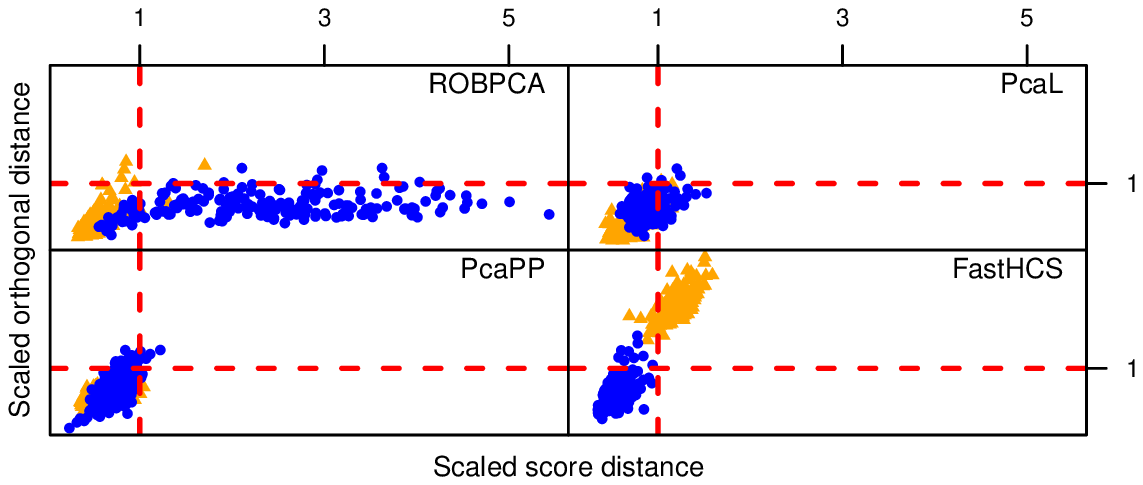}}
\vspace{-1.0em}
\end{figure*}

\begin{figure*}
\centering
\floatbox[{\capbeside\thisfloatsetup{capbesideposition={right,top},capbesidewidth=4cm}}]{figure}[\FBwidth]
{\caption{Spectra of fifty 250mg (light orange) tablets and eighty 80mg (dark blue) tablets for the Tablet data set.}\label{hcs:tablet1}}
{\includegraphics[width=0.65\textwidth]{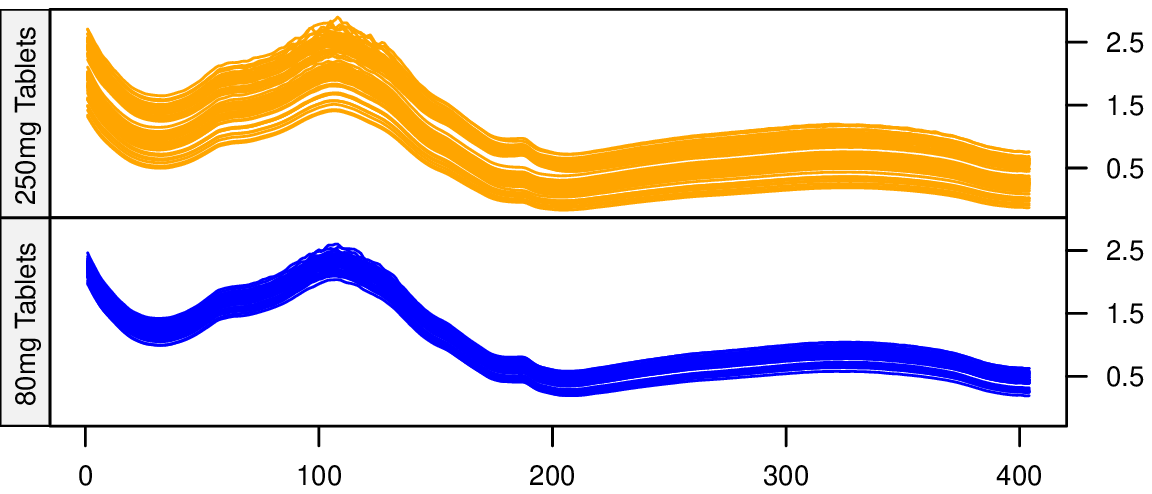}}
\vspace{-1.0em}
\end{figure*} 

\begin{figure*}
\centering
\floatbox[{\capbeside\thisfloatsetup{capbesideposition={right,top},capbesidewidth=4cm}}]{figure}[\FBwidth]
{\caption{Diagnostic plots of the scaled score and orthogonal distances of the measured spectra corresponding to the four robust PCA fits for the Tablet data set. The observations corresponding to 80mg (250mg) tablets are shown as dark blue circles (light orange triangles).}\label{hcs:tablet2}}
{\includegraphics[width=0.3\textwidth,angle=-90]{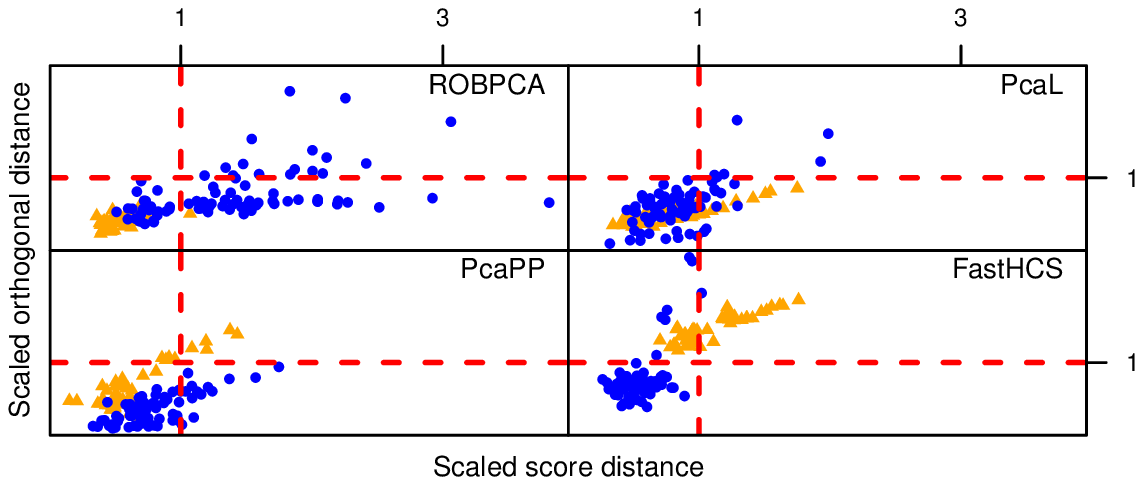}}
\vspace{-1.0em}
\end{figure*} 

\begin{figure*}
\centering
\floatbox[{\capbeside\thisfloatsetup{capbesideposition={right,top},capbesidewidth=4cm}}]{figure}[\FBwidth]
{\caption{The 198 vectors of cytosine methylation 
 $\beta$ values for the DNA alteration data set. The first 85 curves 
(corresponding to observations taken from blood tissues)
are show in the top panel in light orange. The main group (113 curves) 
corresponding to observations taken from non-blood, tissues are shown in the bottom panel in dark blue.}\label{hcs:rd4}}
{\includegraphics[width=0.65\textwidth]{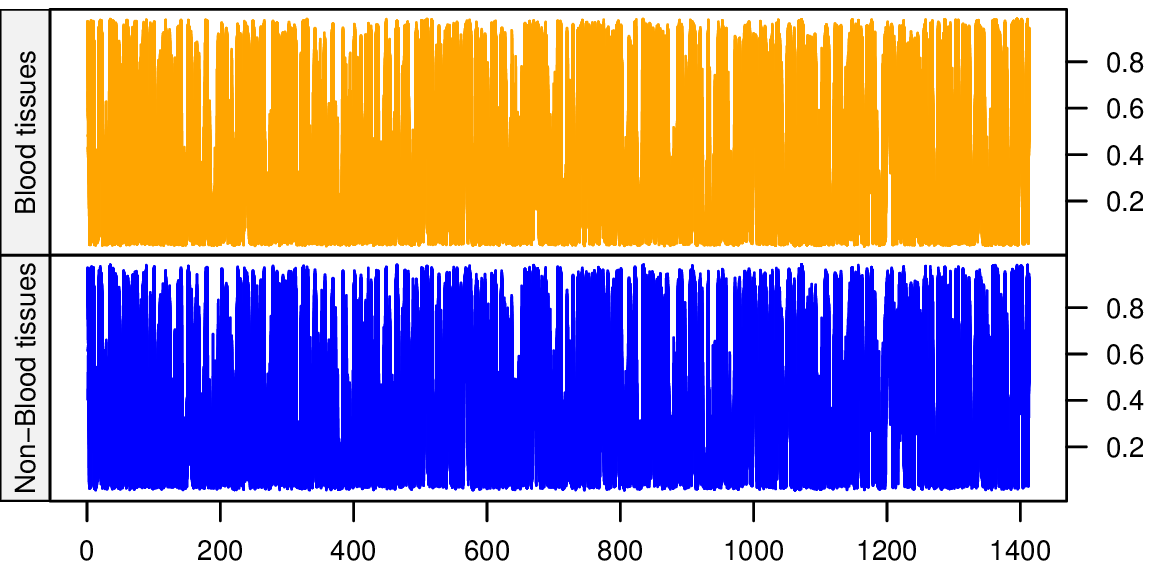}}
\vspace{-1.0em}
\end{figure*} 

\begin{figure*}
\centering
\floatbox[{\capbeside\thisfloatsetup{capbesideposition={right,top},capbesidewidth=4cm}}]{figure}[\FBwidth]
{\caption{Diagnostic plots of the scaled score and orthogonal 
distances of cytosine methylation $\beta$ values corresponding to the four robust PCA algorithms for the DNA alteration data set. The observations with labels
"non-blood" ("blood") are shown as dark blue circles (light orange triangles).}\label{hcs:geneticDiagnostic}}
{\includegraphics[width=0.3\textwidth,angle=-90]{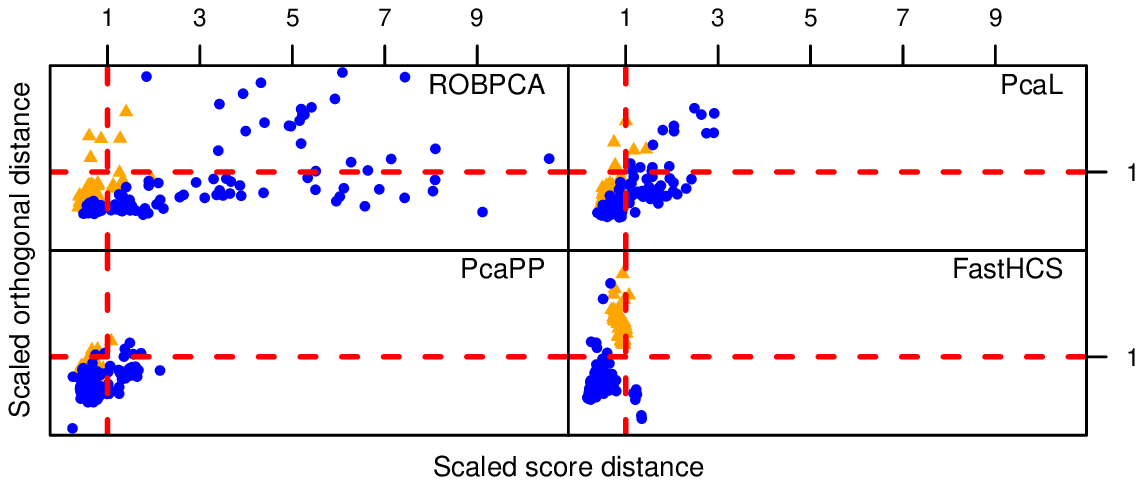}}
\vspace{-1.0em}
\end{figure*} 

\subsection{The Multiple Features data set}
\label{sec:multiplefeatures}

The Multiple Features data set~\citep{hcs:VBal} contains many replications of hand written 
  numerals ('0'-'9') extracted from nine
 original maps of a Dutch public utility. 
 For each numeral, we have 200 replications  
 (the observations) expressed as a vector of 76 
 of Fourier coefficients (the features) describing its shape. Finally, each numeral
 has been manually identified, yielding an extra vector of class labels.  
In this application, we will combine the vectors 
of Fourier coefficients corresponding to the  200 replications of the digit '1' to 
the vector of Fourier coefficients corresponding to the first 150 replications of
the digit '0' (so that $n=350$ and $p=76$). The goal of the methods will be to distinguish the '0's and the '1's.
 
  To give an impression of the differences between the two groups, we plot the Fourier coefficients 
corresponding to the main (outlier) subgroup 
 in the bottom (top) panels of Figure~\ref{hcs:rd1} 
  as dark blue (light orange) curves. 
In general, the curves corresponding to the members of the two groups are visually similar. In particular, the vertical ranges 
of both largely overlap, and both sets of curves 
exhibit a similar pattern of variance clustering 
where the central 40 Fourier coefficients have
systematically less dispersion than higher or lower ones.

Figure~\ref{hcs:rd2} depicts the four resulting diagnostic plots. We assign to each observation a color (dark blue or light orange) and a plot symbol (round or triangle) depending on whether the corresponding curve describes a member of class '1' or '0', respectively. 
The outlier plots of PcaPP and PcaL show that neither method makes any distinction between the two digits, and observations from both groups influence the corresponding PCA models.   
ROBPCA discovers a different structure in the data, confounding the '0's as the majority group and the '1's as outliers on the model space. Since only a few '1's are $\od$  outliers, almost all of the observations influence the fitted loadings. 
In contrast to the other methods, FastHCS correctly identifies all of the '0's as outliers and identifies some '1's that might warrant additional scrutiny.

%\FloatBarrier

\subsection{The Tablet data set}
\label{sec:tablets}

The Tablet data~\citep{hcs:D02} contains the results of an analysis on Escitalopram$^{\mbox{\scriptsize\textregistered}}$ tablets from the pharmaceutical company H. Lundbeck A/S using near-infrared (NIR) spectroscopy. 
The study  includes tablets of four different dosages from pilot, laboratory and full scale production settings are included. Each tablet (the observations) is measured along 404 wavelengths (the variables). 
From this data, we extract two subsets of observations 
which we combine to obtain a new data set formed of two heterogeneous subgroups. Tablets of 80mg will make up the majority group and the rows corresponding to the first 50 tablets with a nominal weight of 250mg will serve as the outliers. This gives a high-dimensional data set (i.e. $p>n$) with $n=130$, $p=404$ and a contamination rate of $\varepsilon=38\%$. 

Figure~\ref{hcs:tablet1}, depicts the spectra of the 250mg (light orange) and 80mg (dark blue) tablets. The spectra of the 250mg tablets follow a different multivariate pattern than those of the 80mg tablets. For example, the spectra of the former are generally lower and more spread out than the spectra of the latter.~\cite{hcs:D02} explain that accurate models for NIR analyses of medical tablets are valuable for quality control purposes, since they are fast, nondestructive, noninvasive, and require little preparation. The goal of the algorithms will be to fit a model to the 80mg tablets, despite the presence in the sample of many 250mg tablets.

Figure~\ref{hcs:tablet2} depicts the diagnostic plots of the scaled outlyingness measures obtained from each of the algorithms. To enhance the 
distinction between the two groups in our data, we show the 80mg tablets as (dark) blue circles and the 250mg tablets as (light) orange triangles.
These results are similar to those we saw when we examined the Multiple Features data set. Again PcaPP and PcaL do not distinguish between the two groups and ROBPCA uses both groups to fit the loadings parameters and confuses the outliers with the majority group on the model space.
 In contrast, the diagnostic plot derived from the FastHCS fit  establishes that the 250mg tablets do not follow the same multivariate patterns as 80mg tablets and, in fact, depart significantly from it. In the plot, we see that FastHCS assigns the outliers high $\od$ values; excluding them from loadings and eigenvalue estimation. It also assigns many of them high $\sd$ values; revealing their distance on the model space.

%\FloatBarrier 
 
\subsection{DNA Alteration data set}
\label{sec:dna}

In our final case study, we examine another high-dimensional data set; the DNA  Alteration data set~\citep{hcs:christ09}. This data set
consists of cytosine methylation $\beta$ values collected at 
1413 autosomal CpG loci (the variables) in a 
sample of 217 non-pathological human tissue specimens (the observations) taken from 10 different anatosites. 
In \cite{hcs:christ09}, the authors show that the tissue samples in this data set form three well separated subgroups. The first of these constitutes all 113 observations corresponding to cytosine methylation $\beta$ values measured on "non-blood, non-placenta" (henceforth, simply "non-blood") tissues.  A second subgroup of data points comprises the 85 cytosine methylation $\beta$ measurements taken on blood tissues.

In this application, we will combine the 113 measurements 
of cytosine methylation $\beta$ values corresponding
to the samples "non-blood" tissue with 85 measurements 
 taken from blood tissues (so that $n=198$ and $p=1413$).
In Figure~\ref{hcs:rd4}, we plot the 1413 $\beta$ values corresponding to each blood (light orange) and non-blood (dark blue) observation.
Visually, the curves of these two groups appear difficult to distinguish from one another. In particular, the vertical range of both overlap and the groups do not exhibit any pronounced difference in the variability of the variables.

The diagnostic plot for PcaPP (Figure~\ref{hcs:geneticDiagnostic}) reveals that the fitted model regards blood and non-blood tissue to be quantitatively similar. ROBPCA and PcaL also detect almost none of the outliers, but additionally consider a number of the non-blood observations to be $\sd$ outliers. As in the previous case studies, FastHCS correctly identifies all of the outliers. As a consequence, the parameter estimates corresponding to this model are more likely to reflect the true structure of non-blood tissue than those fitted by the other algorithms.

\section{Outlook}
In this article we introduced FastHCS to satisfy a number of criteria we expect a robust PCA method to have. In both the simulations and real data examples we performed, FastHCS met all of these criteria. In contrast, state-of-the-art methods did not, and often produced results one would expect from a non-robust method. This may seem like an extreme outcome, but it is in fact the very nature of dealing with outliers: if a method fails to identify them, the resulting model fit is often profoundly changed.

It is interesting to compare the performance of FastHCS and ROBPCA because these methods both use variants of projection pursuit. While FastHCS compares the fit produced by the $I$-index to that from the projection pursuit criterion, ROBPCA relies completely on the projection pursuit criterion to construct its initial subset. Thus, the difference in performance between FastHCS and ROBPCA that we observe in our simulations and real data examples arises from the fact that FastHCS nearly always chooses the $I$-index subset over the projection pursuit  one.

In most applications, admittedly, data settings and contamination
 patterns will not be as difficult as
 those we featured in our simulations and real data examples, and in these easier cases 
 the different methods will, hopefully, concur. 
Nevertheless, in three real data examples from fields where PCA is widely used, we were able to 
establish that real world situations can be challenging enough to push current state-of-the-art outlier detection procedures to their  
limits and beyond, justifying the  
development of better solutions.
In any case, given that in practice we do not
 know the configuration of the outliers, as
 data analysts, we prefer to carry out our inferences
 while planning for the worst contingencies. 

\section{Acknowledgements}
\noindent The authors wish to acknowledge the helpful comments from two anonymous 
referees and the editor which improved this paper. 

%%%%%%%%%%%%%%%%%%%%%%%%%%%%%%%%%%%%%%%%%%%

\FloatBarrier

\begin{appendix}
\section{Vulnerability of the $I$-index to orthogonal outliers} 
\label{sec:IindVul}
Throughout this appendix, let $\pmb Y$ be an $n\times p$ data matrix of uncontaminated observations drawn from 
a rank $q$ distribution $\mathscr{F}$, with $q$ and integer satisfying $2<q<\min(p,n)$.
 However, we do not observe $\pmb Y$ but an $n\times p$ 
  (potentially) corrupted data matrix $\pmb Y^{\varepsilon}$ 
   that consists of $g<n$ observations from $\pmb Y$ and 
    $c=n-g$ arbitrary values with $\varepsilon=c/n$ denoting 
     the (unknown) rate of contamination. Throughout,  
     $h=\lceil\slfrac{(n+q+1)}{2}\rceil$ and the PCA estimates $(\pmb t^I, \pmb L_q^I,\pmb P_q^I)$ are defined as in Section~\eqref{sfhcs} with $(\pmb L_q^I)_j,1\leqslant j\leqslant q$ will denoting the $j$-th diagonal entry of $\pmb L_q^I$. 
     
We will consider the finite sample breakdown~\citep{hcs:D82} in the context of PCA following~\citep{hcs:LC85}: 
\begin{eqnarray}
\varepsilon_1 &=& \min_{1\leqslant c\leqslant n}\{\varepsilon=\frac{c}{n}:  (\pmb L_q)_1=\infty\}\label{expfsbdp}\\
\varepsilon_2 &=& \min_{1\leqslant c\leqslant n}\{\varepsilon=\frac{c}{n}:  (\pmb L_q)_{q}=0\}\label{impfsbdp}
\end{eqnarray}
Equation~\eqref{expfsbdp} defines the so-called finite sample explosion breakdown point and Equation~\eqref{impfsbdp} the 
so-called finite sample implosion breakdown point of PCA estimates $(\pmb t, \pmb L_q,\pmb P_q)$, and the general finite
sample breakdown point is $\varepsilon^*_n = \min(\varepsilon_1, \varepsilon_2)$.

The following assumptions (as per, for example~\cite{ppcs:T94}) all pertain to the original, uncontaminated,  
data set $\pmb Y$. 
We will consider the case whereby the point cloud formed by $\pmb Y$ lies in \textit{general position} in 
$\mathbb{R}^q$. The following definition of \textit{general position} is adapted from 
\cite{mcs:RL87}:

 \bigskip

\textsc{Definition} 1: \textit{General position in $\mathbb{R}^q$}. 
$\pmb Y$ is in general position in $\mathbb{R}^q$ if 
no more than $q$-points of $\pmb Y$
 lie in any $(q-1)$-dimensional 
 affine subspace. 
For $q$-dimensional data, this
means that there are no more than $q$ points
 of $\pmb Y$ on any hyperplane, so that any $q+1$ points of
$\pmb Y$ always determine a $q$-simplex with non-zero determinant.

\bigskip
\noindent 
The $I$-index is shift invariant so that, 
w.l.o.g., we only consider cases where the
good observations are centered at the origin.
Throughout, we will also assume that the 
members of $\pmb Y$ are bounded:
\begin{eqnarray*}\label{eqbdp1}
\max_{i=1}^n||\pmb y_i||<U_0 
\end{eqnarray*}
for some bounded scalar $U_0$ depending only on the uncontaminated observations and that the uncontaminated  
observations contain no duplicates:
\begin{eqnarray*}
||\pmb y_i-\pmb y_j||>0\;\forall\;1\leqslant i<j\leqslant n. 
\end{eqnarray*}

\subsection{Theorem 1: The implosion breakdown, $\varepsilon_2(\pmb t^I, \pmb L_q^I,\pmb P_q^I)$, is $\slfrac{(n-h+1)}{n}$}

\begin{proof}
If at least $h$ rows of $\pmb Y^{\varepsilon}$ are in general position in $\mathbb{R}^q$, any subset of $h$ observations will  contain at least $q+1$ observations in general position. This guarantees that the $q^{th}$ eigenvalue corresponding to any $h$-subset is non-zero~\citep{hcs:Seber}. Thus, it follows
that $\varepsilon_2(\pmb t^I, \pmb L_q^I,\pmb P_q^I)=\slfrac{(n-h+1)}{n}$. 
\end{proof}

\subsection{Finite sample explosion
breakdown of $(\pmb t^I, \pmb L_q^I,\pmb P_q^I)$}
Denote $\pmb z\in\mathbb{R}^p$ 
the outlying entries of $\pmb Y^{\varepsilon}$ and
$\pmb z^m=||\pmb z\pmb P_0^m||$.
The only outliers 
capable of causing explosion 
breakdown must satisfy:
\begin{eqnarray}
||\pmb z||&\geqslant&U_1,\label{eqp1}\\
\min_m||\pmb z\pmb P_0^m||&\leqslant&U_2.\label{eqp2}
\end{eqnarray}
for any bounded scalar $U_1$ and $U_2$ 
depending only on the uncontaminated observations.

\begin{proof}
Suppose that the outliers do not satisfy 
Equation~\eqref{eqp1} so that $\max_{i}||\pmb y_i^{\varepsilon}||\leqslant U_1$, but that the PCA estimates $(\pmb t^I, \pmb L_q^I,\pmb P_q^I)$  break down. 
This leads to a contradiction since 
\begin{eqnarray}
(\pmb L_q^I)_1\leqslant\max_{i\in H^I}||\pmb y_i^{\varepsilon}||
\end{eqnarray}
Therefore, for a contaminated $h$-subset to cause explosion 
 breakdown, the outliers must satisfy 
  Equation~\eqref{eqp1}.

Assume that an outlier $\pmb z$ does not satisfy Condition~\eqref{eqp2}. \cite{SV14}
showed that any $h$ subset $H^m$ 
   indexing $\pmb z$ will have an 
   unbounded value of $I(H^m,\pmb S_0^m)$ if and only if 
   $\pmb z^m$ is unbounded. 
    But for the uncontaminated data, it holds that
\begin{eqnarray} 
\max_i\min_m||\pmb y_i\pmb P_0^m||&\leqslant&U_2
\end{eqnarray}  
 so if the contaminated data set $\pmb Y^{\varepsilon}$ 
contains at least $h$ entries from the original data matrix $\pmb Y$, then it is always possible to construct 
a subset $H^m$ of entries of $\pmb Y^{\varepsilon}$ for which 
$I(H^l,\pmb S_0^l)$ is bounded so that 
$H^m$ will never be selected over $H^l$.

\end{proof}

%%%%%%%%%%%%%%%%%%%%%%%%%%%%%%%%%%%%%%%%%%%%
\section{The finite sample breakdown point of FastHCS}
\label{app:AB}
In this appendix, we derive the finite sample breakdown point of FastHCS.
Define $\pmb Y$, $\pmb Y^\varepsilon$ and $\varepsilon^*_n$ as in Appendix A. 
Recall that 
\begin{eqnarray}
D(\pmb Y^\varepsilon,H^I,H^{PP})=\ave_{j=1}^q\log\frac{\ave_{i \in H^I}((\pmb y^\varepsilon_i-\pmb t^{I})\pmb P^{I}_j)^2}{\var_{i \in H^\bullet}(\pmb y^\varepsilon_i\pmb P^{I}_j)}\nonumber\\
-\max_{j=1}^q\log\frac{\ave_{i \in H^\bullet}((\pmb y^\varepsilon_i-\pmb t^{PP})\pmb P^{PP}_j)^2}{\var_{i \in H^{-}}(\pmb y^\varepsilon_i\pmb P^{PP}_j)},
\label{eq:recriterion}
\end{eqnarray}
where $H^-=H^{PP} \setminus H^{I}$.
Then, if $D(\pmb Y^\varepsilon,H^I,H^{PP})>0$ or if $\displaystyle\max_{j=1}^q\var_{i \in H^{-}}(\pmb y^\varepsilon_i\pmb P^{PP}_j)=0$ then the final FastHCS estimates are based on $H^{PP}$. Otherwise, they are based on $ H^{I}$.

\begin{lemma}
If $||\pmb y^\varepsilon_i||>U_1$ and $\varepsilon < \slfrac{(n-1)}{2n}$, then $i \notin H^{\bullet}$.
\end{lemma}
\begin{proof}
~\citep{hcs:DH09} showed that the population breakdown point of $(\pmb t^{PP}, \pmb L_q^{PP},\pmb P_q^{PP})$ is 50\%, which corresponds to a finite sample breakdown point of $\slfrac{(n-1)}{2n}$. Consequently, $H^{PP}$ will not index any data point for which $||\pmb y^\varepsilon_i||>U_1$. Since $H^{\bullet}$ indexes the overlap between $H^I$ and $H^{PP}$, if $||\pmb y^\varepsilon_i||>U_1$, then $i \notin H^{\bullet}$.
\end{proof}

\begin{lemma}
When $\pmb Y$ is in general position, $n>q>2$, and $\varepsilon < \varepsilon_1 = \slfrac{(n-1)}{2n},\;  (\pmb L_q^I)_1 <\infty$.
\end{lemma}
\begin{proof} 
We will proceed by showing that the denominators in Equation~\eqref{eq:recriterion} are bounded, while only the numerator dependent on $H^{PP}$ is bounded.

 Lemma 1 implies there exists a fixed constant $U_4$ such that
\begin{equation}
||\pmb y^\varepsilon_i\pmb P_j|| < U_4\; \forall\;i \in H^{\bullet},\; 1\leqslant j \leqslant q
\label{eq:boundedVarIn}
\end{equation}
for any orthogonal matrix $\pmb P$. Similarly, since the projection pursuit approach has a breakdown point of $\slfrac{(n-1)}{2n}$, there exists a fixed $U_5$ such that
\begin{equation}
||\pmb y^\varepsilon_i\pmb P_j|| < U_5\; \forall\;i \in H^{PP},\; 1\leqslant j \leqslant q
\label{eq:boundedVarInPP}
\end{equation}
As a consequence of \eqref{eq:boundedVarIn} and \eqref{eq:boundedVarInPP}, there exists a fixed constant $U_6$ such that:
\begin{eqnarray}
\sum_j \log(\var_{i \in H^\bullet}(\pmb y^\varepsilon_i\pmb P^{I}_j))&<&U_6 \label{eq:boundDen}\\
\sum_j \log(\var_{i \in H^{-}}(\pmb y^\varepsilon_i\pmb P^{PP}_j))&<&U_6\nonumber.
\end{eqnarray}
Next, note that
\begin{eqnarray}
\max_j\log(\ave_{i \in H^I}((\pmb y^\varepsilon_i-\pmb t^{I})\pmb P^{I}_j)^2) &=& (\pmb L_q^I)_1 \label{eq:maxEigV}\\
\min_j\log(\ave_{i \in H^I}((\pmb y^\varepsilon_i-\pmb t^{I})\pmb P^{I}_j)^2) &=& (\pmb L_q^I)_q\geqslant\epsilon>0, \label{eq:minEigV}
\end{eqnarray}
(Equation \eqref{eq:minEigV} follows from Appendix A, Theorem 1), so that 
\begin{equation}
\sum_j\log(\ave_{i \in H^I}((\pmb y^\varepsilon_i-\pmb t^{I})\pmb P^{I}_j)^2)
\label{eq:unboundNumI}
\end{equation}
is not bounded from above. Conversely, $(\pmb t^{PP}, \pmb L_q^{PP},\pmb P_q^{PP})$ has an explosion breakdown point of $\slfrac{(n-1)}{2n}$, so that there exists a fixed $U_8$ such that:
\begin{equation}
\sum_j \log(\ave_{i \in H^\bullet}((\pmb y^\varepsilon_i-\pmb t^{PP})\pmb P^{PP}_j)^2)< U_8.
\label{eq:boundNumPP}
\end{equation}

From Equations~\eqref{eq:boundDen} and the unboundedness of~\eqref{eq:unboundNumI} it follows that the  left-hand side in Equation~\eqref{eq:recriterion} is unbounded. However, 
by Equations~\eqref{eq:boundDen} and~\eqref{eq:boundNumPP}, the right-hand side of Equation~\eqref{eq:recriterion} is bounded from above so that in cases where  
outliers cause explosion breakdown of $(\pmb t^I, \pmb L_q^I,\pmb P_q^I)$, criterion \eqref{eq:recriterion} will select $H^* = H^{PP}$.
Since the breakdown point of $(\pmb t^{PP}, \pmb L_q^{PP},\pmb P_q^{PP})$ is $\slfrac{(n-1)}{2n}$, we have that $\varepsilon_1 = \slfrac{(n-1)}{2n}$.
\end{proof}

\begin{lemma} When $\pmb Y$ is in general position, $n>q>2$, and $\varepsilon < \varepsilon_2 = (n-h+1)/n$, then $ (\pmb L_q^I)_q > 0$. \end{lemma}
\begin{proof}
By Appendix A, Theorem 1, we have that the implosion breakdown point of $(\pmb t^I, \pmb L_q^I,\pmb P_q^I)$ is $\slfrac{(n-h+1)}{n}$. The implosion breakdown point of $(\pmb t^{PP}, \pmb L_q^{PP},\pmb P_q^{PP})$ is $\slfrac{(n-1)}{2n}$, which is higher, so it follows that $\varepsilon_2=\slfrac{(n-h+1)}{n}$. 
\end{proof}

\begin{theorem}
For $n>p+1>2$, the finite sample breakdown point of  $\pmb L_q$ is 
\begin{equation}
\varepsilon_n^*(\pmb L_q,\pmb Y^\varepsilon)=\slfrac{(n-h+1)}{n}. \nonumber
\label{eq:eigenbreakdown}
\end{equation}
\end{theorem}
\begin{proof}
The finite sample breakdown point of $\pmb L_q = \min(\varepsilon_1, \varepsilon_2)$. Given Lemmas 2 and 3,  $\min(\slfrac{(n-1)}{2n}, (n-h+1)/n) = (n-h+1)/n$.
\end{proof}

\section{Measures of dissimilarity for robust PCA fits.}
The objective of the simulation studies in Section~\ref{mcs:s5} is to measure how much the fitted PCA parameters $(\pmb t,\pmb L_q^{},\pmb P_q^{})$ obtained by four robust PCA methods deviate from the true 
$(\pmb \mu^u,\pmb\Lambda_q^{u},\pmb \Pi_q^{u})$ when they are exposed to outliers.
One way to compare PCA fits is with respect to their eigenvectors, as in the \emph{maxsub} criterion~\citep{BG73}:
\begin{equation} 
\text{maxsub}(\pmb P_q)=\text{arccos}(\lambda_q^{\slfrac{1}{2}}(\pmb D_q)),\nonumber 
\end{equation}
where $\lambda_q(\pmb D_q)$ is the smallest eigenvalue of the matrix $ \pmb D_q^{}=\pmb\Pi_q^\top\pmb P_q^{}\pmb P_q^\top \pmb\Pi_q^{}$.
The maxsub has an appealing geometrical interpretation as it represents the maximum angle between a vector in $\pmb \Pi_q$ and the vector 
most parallel to it in $\pmb P_q$. However, it does not exhaustively account for the dissimilarity between two sets of eigenvectors. 
As an alternative to the \emph{maxsub}, \cite{K79} proposes the total dissimilarity:
\begin{equation}\label{spher} 
\text{sumsub}(\pmb P_q)=\sum_{j=1}^q\lambda_j(\pmb D_q),
\end{equation}
which is an exhaustive measure of dissimilarity for orthogonal matrices.
Furthermore, because $\sum_{j=1}^q\lambda_j(\pmb D_q)=\Tr(\pmb D_q)$ and $|\pmb D_q|=1$~\citep{K79}, it is readily 
seen that~\eqref{spher} is a measure of sphericity of $\pmb D_q$ (it is proportional to the likelihood
ratio test statistics for non-sphericity of $\pmb D_q$~\citep[p. 333-335]{m82}). 
However, note that~\eqref{spher} now forfeits the geometric interpretation enjoyed by the \emph{maxsub}.

 In any case, measures of dissimilarity based solely on eigenvectors, such as the \emph{maxsub} 
  or \emph{sumsub}, necessarily fail to account for bias in the estimation of the eigenvalues. 
 This is problematic when used to evaluate robust
 fits because it is possible for outliers to exert substantially more influence on $\pmb L_q$ than on $\pmb P_q$. 
 An extreme example is given by the so-called good leverage type of contamination in which the outliers lie on the subspace spanned by $\pmb \Pi_q$ so that even the classical PCA estimate  (whose eigenvalues can be made arbitrarely bad by such outliers) will have low values of $\text{maxsub}(\pmb P_q)$.
 
  In contrast, we are interested in an exhaustive measure of dissimilarity; one that summarizes the the effects of the outliers on all the parameters of the PCA fit into a single number, so that the algorithms can be ranked in terms total dissimilarity. To construct such a measure, it is logical to base it on $\pmb \varSigma_q^u=\pmb \Pi_q^{u}\pmb \Lambda_q^{u}(\pmb \Pi_q^{u})^{\top}$ and its estimate $\pmb V_q=\pmb P_q^{}\pmb L_q^{}\pmb P_q^{\top}$ because they contain
  all the parameters of the fitted model. For our purposes, one need to only consider the effects of outliers on
 $\pmb G_q=|\pmb V_q|^{-1/q}\pmb V_q$, the shape component of $\pmb V_q$~\citep{r14}. This is because to rank
 the observations in a contaminated sample in terms of their true outlyingness (and thus reveal the outliers), 
 it is sufficient to estimate the shape component of $\pmb \varSigma_q^u$ correctly. 
  Consequently, an exhaustive measure of dissimilarity between $\pmb G_q$ and $\pmb \Gamma_q=|\pmb \varSigma_q|^{-1/q}\pmb \varSigma_q$
 is given by $\phi((\pmb \Gamma^u_q)^{-1/2}\pmb G_{q}(\pmb \Gamma^u_q)^{-1/2})$, where $\phi$ is any measure of non-sphericity of its argument. In practice 
 several choices of $\phi$ are possible, the simplest being the condition number of $\pmb W$ which is defined as the ratio of the largest to the smallest eigenvalue of $\pmb W$~\citep{MY95}, explaining the definition of $\text{bias}(\pmb V_q)$. 
\end{appendix}

\vspace{-3mm}
\bibliographystyle{spmpsci}

\begin{thebibliography}{00}
\bibitem[Bj\"{o}rck and Golub (1973)]{BG73}
Bj\"{o}rck, \r{A}. and Golub, G. H. (1973). Numerical Methods for Computing Angles Between Linear Subspaces. Mathematics of Computation, 27, 2, 579--594.
\bibitem[Christensen et al. (2009)]{hcs:christ09}
Christensen,  B.C  Houseman, E.A. Marsit, C.J. Zheng, S. Wrench, M.R. Wiemels, J.L. 
Nelson, H.H. Karagas, M.R. Padbury, J.F. Bueno, R. Sugarbaker, D.J Yeh, R., Wiencke, J.K. 
Kelsey, K.T. (2009).
Aging and Environemental Exposure Alter Tissue-Specific DNA Methylation Dependent upon CpG Island Context. PLoS Genetics 5(8), e1000602.
\bibitem[Croux and Ruiz-Gazen (2005)]{hcs:cr05}
Croux, C. and Ruiz-Gazen, A. (2005). 
High breakdown estimators for principal components: the projection-pursuit approach revisited. Journal of Multivariate Analysis, 95, 206--226.
\bibitem[Donoho (1982)]{hcs:D82}
Donoho, D.L. (1982). 
Breakdown properties of multivariate location estimators.
Ph.D. Qualifying Paper Harvard University.
\bibitem[Debruyne and Hubert (2009)]{hcs:DH09}
Debruyne, M. and Hubert, M. (2009).
The influence function of the Stahel-Donoho covariance estimator of smallest outlyingness.
Statistics \& probability letters 79(3), 275--282.
\bibitem[Deepayan (2008)]{mcs:D08}
Deepayan, S. (2008). 
Lattice: Multivariate Data Visualization with R. Springer, New York.
\bibitem[Dyrby et al. (2002)]{hcs:D02}
Dyrby, M. Engelsen, S.B. N\o rgaard, L. Bruhn, M. and Lundsberg Nielsen, L. (2002).
Chemometric Quantitation of the Active Substance in a Pharmaceutical Tablet Using Near Infrared (NIR) Transmittance and NIR FT Raman Spectra
Applied Spectroscopy 56(5): 579--585 .
\bibitem [Hubert et al. (2005)]{hcs:HRV05}
Hubert, M.  Rousseeuw, P. J. and Vanden Branden, K. (2005).
ROBPCA: a new approach to robust principal components analysis. 
Technometrics, 47, 64--79.
\bibitem [Hubert et al.(2014)]{r14}
Hubert, M., Rousseeuw, P. and Vakili, K. (2014).
Shape bias of robust covariance estimators: an empirical study.
Statistical Papers, Volume 55, Issue 1, pp 15--28.
\bibitem[Jensen (1986)]{hcs:J86}
Jensen, D. R. (1986), The Structure of Ellipsoidal Distributions, II. Principal Components. Biometrical Journal, 28: 363--369.
\bibitem[Jolliffe (2002)]{Jolliffe2002}
Jolliffe, I.T. (2002).
Principal Component Analysis. 
Springer, New York. Second Edition.
\bibitem[Krzanowski (1979)]{K79}
Krzanowski, W.J. (1979). 
Between-Groups Comparison of Principal Components.
Journal of the American Statistical Association, Vol. 74, No. 367, 
pp. 703--707.
\bibitem [Li and Chen (1985)]{hcs:LC85}
 Li, G., Chen, Z. (1985). Projection-pursuit approach to robust dispersion matrices and principal components: primary
theory and Monte Carlo. Joural of the American Statistical Association, 80, pp. 759--766.
\bibitem [Locantore et al. (1999)]{hcs:L99}
Locantore, N., Marron, J. S., Simpson, D. G., Tripoli, N., Zhang, J. T., and Cohen, K. L. (1999). 
Robust principal component analysis for functional data.
Test. 8(1), 1--73.
\bibitem[Maronna and Yohai (1995)]{MY95}
 Maronna R. A. and Yohai  V.J. (1995).
 The Behavior of the Stahel-Donoho Robust Multivariate Estimator. Journal of the American Statistical Association 90 (429), 330--341.
\bibitem [Maronna (2005)]{hcs:M05}
Maronna, R. (2005). 
Principal Components and Orthogonal Regression Based on Robust Scales. Technometrics, 47, 264--273.
\bibitem[Maronna et al. (2006)]{mcs:MMY06}
Maronna, R. A.; Martin, R. D. and Yohai, V. J. (2006).
Robust Statistics: Theory and Methods. Wiley, New York.
\bibitem[Muirhead (1982)]{m82}
Muirhead, R.J. (1982). 
Aspects of Multivariate Statistical Theory.
John Wiley and Sons, New York.
\bibitem[R Core Team (2012)]{rcore}
R Core Team (2014).
R: A Language and Environment for Statistical Computing.
R Foundation for Statistical Computing. Vienna, Austria.
\bibitem[Rousseeuw and Leroy (1987)]{mcs:RL87}
Rousseeuw, P.J. and Leroy, A.M. (1987). 
Robust Regression and Outlier Detection. Wiley, New York.
\bibitem [Schmitt et al. (2014)]{SV14}
Schmitt, E. \"Ollerer, V. and Vakili, K. (2014).
The finite sample breakdown point of PCS. Statistics \& Probability Letters, 94, 214--220.
\bibitem[Seber (2008)]{hcs:Seber}
Seber,  G. A. F. (2008).
Matrix Handbook for Statisticians. 
Wiley Series in Probability and Statistics. Wiley, New York.
\bibitem[Stahel (1981)]{hcs:S81} 
Stahel W. (1981). Breakdown of Covariance Estimators. Research Report 31, Fachgrupp f\"{u}r Statistik, E.T.H. Z\"{u}rich.
\bibitem[Todorov and Filzmoser (2009)]{mcs:TF09}
Todorov V. and Filzmoser P. (2009).
An Object-Oriented Framework for Robust Multivariate Analysis.
Journal of Statistical Software, 32, 1--47.
\bibitem[Tyler (1994)]{ppcs:T94}
Tyler, D.E.  (1994). 
Finite Sample Breakdown Points of Projection Based Multivariate Location and Scatter Statistics.
\bibitem[Vakili and Schmitt (2014)]{hcs:VS13}
Vakili, K. and Schmitt, E.  (2014).
Finding multivariate outliers with FastPCS.
Computational Statistics \& Data Analysis, Vol. 69, 54--66.
\bibitem [Van Breukelen et al. (1998)]{hcs:VBal}
Van Breukelen, M. Duin, R.P.W. Tax, D.M.J. and  Den Hartog, J.E. (1998).
Handwritten digit recognition by combined classifiers. Kybernetika, 34, 381--386.
\bibitem[Wu et al. (1997)]{hcs:WMJ97}
Wu, W., Massart, D. L., and de Jong, S. (1997), The Kernel PCA Algorithms
forWide Data. Part I: Theory and Algorithms. Chemometrics and Intelligent
Laboratory Systems, 36, 165--172.
\bibitem[Yohai and Maronna (1990)]{mcs:YM90}
Yohai, V.J. and Maronna, R.A. (1990).
The Maximum Bias of Robust Covariances.
Communications in Statistics--Theory and Methods, 19, 2925--2933.
\end{thebibliography}

\end{document}